\documentclass[11pt]{scrartcl}
\usepackage[left=2cm, right=2cm, top=2cm]{geometry}
\usepackage{booktabs}

\usepackage[utf8]{inputenc}
\usepackage[english]{babel}

\usepackage{amsmath,amssymb}
\usepackage{amsthm}
\usepackage{amssymb}
\usepackage{mathtools}
\usepackage{multicol}
\usepackage{enumitem}

\usepackage[table]{xcolor} 
\usepackage{subcaption}

\usepackage[colorinlistoftodos,prependcaption,textsize=tiny]{todonotes}
\usepackage{xargs} 
\newcommandx{\add}[2][1=]{\todo[linecolor=red,backgroundcolor=red!25,bordercolor=red,#1]{#2}}
\newcommandx{\change}[2][1=]{\todo[linecolor=blue,backgroundcolor=blue!25,bordercolor=blue,#1]{#2}}
\newcommandx{\info}[2][1=]{\todo[linecolor=OliveGreen,backgroundcolor=OliveGreen!25,bordercolor=OliveGreen,#1]{#2}}

\theoremstyle{plain}
  \newtheorem{thm}{Theorem}
  \newtheorem{lm}{Lemma}
  
  \newtheorem{fct}{Fact}

\theoremstyle{definition}
  \newtheorem{df}{Definition}

  \newtheorem{rmk}{Remark}

\newcommand{\psv}{\mathcal{P}^\dagger(\mt{Val})}
\newcommand{\mt}[1]{\mathtt{#1}}
\newcommand{\mb}[1]{\mathtt{#1}}
\newcommand{\mc}[1]{\mathcal{#1}}
\newcommand{\jt}{\rightarrow}
\newcommand{\n}{\neg}
\newcommand{\et}{\wedge}

\newcommand{\B}{\mbox{$\Box$}}

\newcommand{\vp}{\mbox{$\varphi$}}

\newcommand{\gn}[1]{\mbox{$\ulcorner #1 \urcorner$}}
\newcommand{\ov}[1]{\overline{#1}}
 \newcommand{\ds}{\mathtt{D}}
 \newcommand{\nds}{\overline{\mathtt{D}}}

\usepackage{hyperref}
\hypersetup{
  colorlinks=true,
  linkcolor=black,
  citecolor=black,
  urlcolor=blue
}

\newcommand{\publishednotice}{%
  \begin{center}
    \fbox{%
      \parbox{0.92\textwidth}{\small
        \textbf{Notice.} This is a draft version of a paper that has been published. For citation and reference, please consult the published version:

        P.\ Pawlowski, ``T\,-BAT Semantics and Its Logics,'' \textit{Logique et Analyse} \textbf{264} (2023), 335--356.
        DOI: \href{https://doi.org/10.2143/LEA.264.0.3294241}{10.2143/LEA.264.0.3294241}.
      }%
    }
  \end{center}
  \vspace{1em}
}

\date{}
\title{T\,-BAT semantics and its logics}

\begin{document}
\maketitle

\publishednotice

\begin{abstract}
\textbf{T-BAT} logic is a formal system designed to express the notion of informal provability. This type of provability is closely related to mathematical practice and is quite often contrasted with formal provability, understood as a formal derivation in an appropriate formal system. \textbf{T-BAT} is a non-deterministic four-valued logic. The logical values in \textbf{T-BAT} semantics convey not only the information whether a given formula is true but also about its provability status.

The primary aim of our paper is to study the proposed four-valued non-deterministic semantics. We look into the intricacies of the interactions between various weakenings and strengthenings of the semantics with axioms that they induce. We prove the completeness of all the logics that are definable in this semantics by transforming truth values into specific expressions formulated within the object language of the semantics.  Additionally, we utilize Kripke semantics to examine these axioms from a modal perspective by providing a frame condition that they induce. The secondary aim of this paper is to provide an intuitive axiomatization of \textbf{T-BAT} logic.
\end{abstract}
\section{Motivations}

The topic of provability has captivated philosophers for a substantial period of time. A contributing factor to this interest is the existence of at least two distinct notions of provability: formal and informal.\footnote{For a fantastic primer on this debate, have a look at: \cite{AntonuttiMarfori2010,AntonuttiMarfori2018,Leitgeb2009,Shapiro1985}.} The essence of formal provability is intrinsically syntactical, tethered to a bespoke formal language and axiomatic framework. According to this formulation, a mathematical statement is considered provable within a specific formal system iff it can be formally derived within that system i.e. if there is a finite sequence of formulas of the formal language, where each element of this sequence is either an axiom of the formal system, or is obtained by an application of one of the rules of the formal system, to the previous elements, and where the last element in this sequence is the formula in question. The inferential behavior of this provability is captured by a modal logic \textbf{GL}. \textbf{GL} is a propositional modal logic where the intended interpretation of $\B\vp$ is `'formula $\vp$ is formally provable''. Since formal provability is a precise mathematical notion, \textbf{GL} not only captures the formal provability but also this fact has been proven (see \cite{Solovay1976}). To state exactly the relation between $\Box$ and an arithmetical representation of the derivability relation $\vdash$  we start with a sufficiently strong arithmetical theory \textbf{T}. By sufficiently strong, we mean a theory that is capable of expressing its own syntax\footnote{Think about it as any extension of the Robinson's arithmetic \textbf{Q}.} This theory can also represent its own provability predicate $\mathbf{P}_{\mathbf{T}}$. Next key element is a function $\mathbf{r}: \mathtt{Var}_{\mathcal{L}_{\mathbf{GL}}} \mapsto \mathcal{L}_{\mathbf{T}}$ called arithmetical realization. This function tells us how propositional variables are going to  mapped on arithmetical sentences. Relative to this function we define the \textbf{T}-interpretation of a model formula by:

\begin{align}
\mathbf{T}_{\mathbf{r}}(p) = & \mathbf{r}(p) \\
\mathbf{T}_{\mathbf{r}}(\varphi \rightarrow \psi) = &\mathbf{T}_{\mathbf{r}}(\varphi) \rightarrow \mathbf{T}_{\mathbf{T}}(\psi)  \\
\mathbf{T}_{\mathbf{r}}(\varphi \land \psi) = &\mathbf{T}_{\mathbf{r}}(\varphi) \land \mathbf{T}_{\mathbf{T}}(\psi)  \\
\mathbf{T}_{\mathbf{r}}(\varphi \vee \psi) = &\mathbf{T}_{\mathbf{r}}(\varphi) \vee \mathbf{T}_{\mathbf{T}}(\psi)  \\
\mathbf{T}_{\mathbf{r}}(\neg\varphi ) = &\neg\mathbf{T}_{\mathbf{r}}(\varphi)   \\
\mathbf{T}_{\mathbf{r}}(\Box\varphi) = &\mathbf{T}_{\mathbf{r}}(\mathbf{P}_{\mathbf{T}}\ulcorner(\mathbf{T}_{\mathbf{r}}(\varphi))\urcorner) \\
\end{align}
where $\ulcorner \varphi \urcorner$ is the G\"odel number of the formula $\varphi$. By $\varphi_{\mathbf{T}_{\mathbf{r}}}$ we denote the set of all possible \textbf{T}-realizations of the formula $\varphi$. What Solovay shown is the following:

\[\vdash_{\mathbf{GL}} \varphi \text{ iff } \vdash_{\mathbf{T}} \varphi_{\mathbf{T}_{\mathbf{r}}}\]

where $\vdash_{\mathbf{GL}}$ is given by the following set of axioms and rules:

\begin{multicols}{2}
\begin{enumerate}
\item Axioms of propositional logic.
\item \textsf{K}: $\B(\vp \jt \psi) \jt (\B\vp \jt \B\psi)$.
\item \textsf{4}: $\B\vp \jt \B\B\vp$.
\item L\"ob's axiom: $\B(\B\vp \jt \vp) \jt \B\vp$.
\item Rule of \emph{modus ponens:} from $\vp\jt \psi,\vp$ infer $\psi$.
\item Rule of necessitation: if $\vp$ is a theorem so is $\B\vp$.
\end{enumerate}
\end{multicols}

In the other corner of the ring there is the notion of informal provability. This notion is closely related to mathematical practice. Ideally, it supposed to characterize and explain the ways ordinary mathematicians prove theorems since they do not conduct proofs in a fully formal framework. Mathematicians justify their theorems by using informal proofs--- logically connected sequences of sentences spelled out in a natural language expanded with appropriate mathematical notation and vocabulary. The connection between steps in such a proof is not purely syntactical. It seems to be established by a combination of semantical and pragmatical components and it often relies on the mathematical competence and intuition of the reader to be able to fill in the gaps in the proofs. Quite often the rules of inference in these proofs go beyond what is expressible by means of logic for instance by incorporating drawing, pictures, or diagrams.\footnote{For various approaches to informal provability check: \cite{Myhill1960,Shapiro1985,Horsten1998,Priest2006,Leitgeb2009,Pawlowski2023}.}

The relation between these two types of provability is especially tricky. Roughly, there are two camps to be associated with. The proponents of the first camp claim that each informal proof if actually an incomplete sketch of a formal derivation. So, in principle informal proofs are reducible to formal derivations.\footnote{Check \cite{AntonuttiMarfori2010,Rav1999,Azzouni2004} for various stances on this debate.} The proponents of the other camp claim that there are substantial differences between informal and formal proofs. One of the differences that is key for this paper is the difference in the inferential behavior of both types of provability. Proponents of the distinction between these types of proofs claim that the reflection schema:
\[\B\vp \jt \vp\]
is a valid schemata for informal provability and it is not valid in \textbf{GL}. Moreover, it is impossible to add all instances of this schemata together with other intuitive principles for provability, since the resulting first-order arithmetical theory will be inconsistent.

One approach aiming at circumventing this problem by using non-classical logic has been developed in \cite{Pawlowski2022a,Pawlowski2018a,Pawlowski2018}. Their idea was to express informal provability within propositional modal logic. They developed three different logics for this task: \textbf{BAT,CABAT} and \textbf{T-BAT}. The first two are based on this idea that one can divide the set of all mathematical claims into three disjoint subsets. Those that are provable, those that are refutable, and those which are neither provable nor refutable. Each of this statutes correspond to a logical value.

The idea behind \textbf{T-BAT} is more fine grained. The crucial observation is to separate truth-values with the provability status of a given mathematical claim. This results in splitting the set of the claims that  are neither provable, nor refutable into these that are true and these, that are false. The authors of \cite{Pawlowski2022a} proposed a non-deterministic semantics for this logic. They also have stated that this logic is the same as $\mathbf{S4}^-$ from \cite{Omori2016} and based on this observation they provided an axiomatization of \textbf{T-BAT}. Unfortunately, this was a bit too hasty and the proposed axiomatization is not strongly sound and complete with respect to the semantics. In the next section we are going to present in details the non-deterministic semantics of \textbf{T-BAT}. We will explore all the definable logics within this context.

\section{Non-deterministic semantics and \textbf{T-BAT}}

\subsection{History}
In the 1930s, two pioneers, Zich \cite{Zich1938} and Zawirski \cite{Zawirski1936}, independently developed non-deterministic semantics. Zich's contributions concentrated on the philosophical implications of complex-valued propositional calculus in the setting of natural language analysis. Meanwhile, Zawirski's insights were more oriented towards the logical aspects of non-deterministic semantics, as he explored the connections between many-valued logic, probability theory, and the philosophy of quantum mechanics, drawing inspiration from the work of Reichenbach \cite{Reichenbach1935, Reichenbach1936}.

Despite the novelty of their work, their contributions of Zich and Zawirski initially received limited recognition. It was not until the 1960s that their ideas get some academic traction, largely due to the efforts of Rescher. In his seminal work \cite{Rescher1962}, Rescher employed non-deterministic truth-functions to investigate implication in natural language, effectively building upon the theoretical foundations established by Zich and Zawirski. Rescher's research played a crucial role in bringing non-deterministic semantics back into the spotlight and solidifying its position as a significant area of inquiry within the field of logic and language.

However, the impact of Rescher's work was not immediately felt, and non-deterministic semantics remained a dead subfield for several decades. It wasn't until the modern development of these semantics in the framework of modal logic that the true potential of nondeterministic approaches was realized. Kearns \cite{Kearns1981} and Ivlev \cite{Ivlev1988} applied non-deterministic semantics to characterize modal logics. Kearns' motivation was primarily philosophical. He aimed to construct a semantics for normal modal logics that did not rely on the possible-worlds framework. Ivlev's goals, however, were slightly different. He was interested in non-normal modal logics, i.e., logics in which the rule of necessitation (if $\vp$ is a theorem of the system, then so is $\B\vp$) is invalid. It is well known that such systems do not support possible-worlds semantics. Therefore, Ivlev employed non-deterministic semantics to characterize some of these systems.

The modern development of non-deterministic semantics gained significant momentum with Avron's introduction of the non-deterministic matrix (Nmatrix) concept \cite{Avron2011,Avron2005a, Avron2005b, Avron2005c,Avron2005e}.\footnote{Independently, a notion of non-deterministic semantics has been used in the context of tractable inferences \cite{Crawford1998}. It appears that the authors coined the term "non-deterministic semantics," as in previous work this semantics was known as "quasi-truth functional." However, the authors did not introduce the notion of an Nmatrix.} Avron formally introduced the term "Nmatrix" and developed it into a rigorous mathematical concept \cite{Avron2005e}, establishing non-deterministic semantics as a powerful tool for studying proof theory and analyzing paraconsistent logics. Avron and his collaborators extensively developed the theoretical framework through a series of influential papers \cite{Avron2007, Avron2008a, Avron2009a}, particularly focusing on canonical systems, sequent calculi, and cut elimination.

A significant breakthrough in the application of non-deterministic semantics came in the field of modal logic. The groundwork was laid by Kearns \cite{Kearns1981} and Ivlev \cite{Ivlev1988}, who pioneered the use of truth values to convey information about a proposition's truth, necessity, and possibility. This approach was substantially advanced by contemporary logicians \cite{Coniglio2015,Coniglio2016, Omori2016,Pawlowski2021a,Pawlowski2023a,Omori2024,Omori2021,Omori2020,Pawlowski2022a,Graetz2022,Lahav2022,Ferguson2022,Coniglio2019b,Coniglio2021,Coniglio2024,Pawlowski2024,Pawlowski2024g}.

\subsection{Non-deterministic semantics}

Let's begin with the technical foundation. We will work with the language $\mathcal{L}= \{\mathtt{Var}, \n,\jt, \B\}$, where $\mathtt{Var}$ is an countable infinite set of variables and $\n,\jt,\B$ are the connectives. We assume that other connectives are defined as per standard conventions. By $\mc{FR}_{\mathcal{L}}$ we represent the set of all formulas of $\mathcal{L}$. Greek lower case letters such as $\varphi,\psi, \dots$ are used as meta-variables for formulas, whereas Greek capital letters like $\Gamma,\Sigma$ are used for sets of formulas. The key notion of non-deterministic semantics is a notion of a nmatrix.

\begin{df}[Nmatrix]
An \emph{Nmatrix}  $\mathtt{M}= (\mathtt{Val},\mathtt{D},\mathtt{O})$ for the language  $\mathcal{L}$ is at triple:

  \begin{itemize}
  \item $\mathtt{Val}\neq\emptyset$ is a set of truth values.
  \item $\emptyset \neq \mathtt{D}\subseteq \mathtt{Val}$ is a set of designated values. Per convention, we use $\nds$  as  $\nds = \{x \mid x\in \mathtt{Val} \et x \notin \ds \}$ 
  \item   $\mathtt{O}$ is a set  of functions, such that:
 \begin{enumerate}
\item $\ov{\circ}:\mt{Val} \mapsto \psv $, where $\psv$ is the power set of $\mt{Val}$ without the empty set and $\circ \in \{\n,\B \}$.
\item $\ov{\jt}:\mt{Val}^2 \mapsto \psv$. 

\end{enumerate}
 
  \end{itemize}
\end{df}

\begin{df}[Valuation]
Let $\mathtt{M}$ be a nmatrix. A valuation $v$ in an nmatrix $\mathtt{M}$ is a function $v: \mc{FR}_{\mathcal{L}} \mapsto \mathtt{Val}$ such that:
\begin{enumerate}
\item $\forall_{p \in \mathtt{Var}} v(p)\in \mathtt{Val}$. 
\item $\forall_{\vp \in \mc{FR}_{\mathcal{L}}}\, v(\circ(\vp_1))\in \ov{\circ}(\vp_1)$, $\circ \in \{\n,\B\}$.
\item $\forall_{\vp_1,\vp_2 \in \mc{FR}_{\mathcal{L}}} \,  v(\jt(\vp_1,\vp_2))\in \ov{\jt}(\vp_1,\vp_2)$.

\end{enumerate}

\end{df}
We abstain from using prefix notation and instead of $\jt(\vp_1,\vp_2)$ we will write $\vp_1 \jt \vp_2$.

\begin{df}[Tautology]
Let $\mt{M}=\langle \mt{Val},\mt{D},\mt{O}  \rangle$  be  an nmatrix. We say that $\vp$ is a tautology in $\mt{M}$ iff for any $\mt{M}$ valuation $v$, we have $v(\vp)\in \mt{D}$. We use $\vDash_{\mt{M}} \vp$ to denote that fact. 
\end{df}

\begin{df}[Consequence relation]
Let $\mt{M}=\langle \mt{Val},\mt{D},\mt{O}  \rangle$ be an nmatrix. We say that $\varphi$ \emph{follows from} $\Gamma$ or that the inference from $\Gamma$ to $\varphi$ is \emph{valid} (notation $\Gamma \vDash_{ \mt{M}} \varphi$) iff for all valuation $v$, if $v(\psi) \in \mt{D}$ for all $\psi \in \Gamma$, then $v(\vp) \in \mt{D}$.
\end{df}

\section{T-BAT: A Logic for Informal Provability}

\textbf{T-BAT} is a four-valued non-deterministic logic designed to capture the behavior of informal mathematical provability and it was introduced in \cite{Pawlowski2022} as an update of the \textbf{CABAT} system \cite{Pawlowski2018,Pawlowski2021}.

The ultimate aim behind the creation of \textbf{T-BAT} logic is to develop a theory of informal provability that mirrors the methodological approach used in Kripke's theory of truth. The challenges associated with informal provability closely resemble those encountered with truth predicates. In truth-theoretic frameworks, for instance, inferential principles like Tarski’s biconditionals cannot be consistently added to a sufficiently strong first-order arithmetical theory. Kripke’s solution was to use a well-motivated non-classical logic in the background to address this inconsistency. Analogously, here, while we aim to establish the principles of modal logic \textbf{S4} to capture informal provability, directly adding them would lead to inconsistency. Thus, a well-founded non-classical logic is required as the background for a first-order arithmetical theory of informal provability. According to \cite{Pawlowski2022}, this background logic is \textbf{T-BAT}. This paper focuses on its axiomatization.

Having said that, we are going to flesh-out the philosophical idea behind the \textbf{T-BAT}. The semantics of \textbf{T-BAT} are intended to model the perspective of a practicing mathematician. From this standpoint, mathematical claims can be provable, refutable, or neither provable nor refutable. However, claims that fall into this third category—those that are neither provable nor refutable—can still be either true or false. Naturally, if a sentence is provable, it must be true; if a sentence is refutable, it must be false. This intuition yields four logical possibilities:

\begin{itemize}
\item $\mb{P}$  stands for a claim being true and provable.
\item $\mb{t}$  stands for a claim being true and independent (independent = not provable, not refutable).
\item $\mb{f}$  stands for a claim being false and independent.
\item $\mb{R}$ stands for a claim being false and refutable. 
\end{itemize}

The designated values are $\mb{P}$ and $\mb{t}$, capturing the intuition that being a tautology of the system is about truth-preservation. More formally, this intuition is pinned-down by the following notion of a \textbf{T-BAT} Nmatrix. \footnote{Our truth-tables for \textbf{T-BAT} logic are slightly different than the original ones. The differences is in the case for implication for $\mb{f,R}$. In the original formulation of \textbf{T-BAT} the output is $\{\mb{P,t}\}$ whereas in our case it is $\mb{t}$. Given the motivations provided in \cite{Pawlowski2022} we think it is simply a typo. The updated \textbf{T-BAT} semantics is actually the same as of $\mathbf{S4}^-$ from \cite{Omori2016}.}

\begin{df}[\textbf{T-BAT} Nmatrix] 
The \textbf{T-BAT} Nmatrix $\mathtt{M}_{\mb{T}} = (\mathtt{Val}_\mb{T},\mathtt{D}_\mb{T},\mathtt{O}_\mb{T})$ is defined by the following stipulations:

\begin{enumerate}
\item  $\mathtt{Val}_\mb{T}=\{\mb{P,t,f,R}\}$.
\item  $\mathtt{D}_\mb{T}=\{\mb{P,t}\}$.
\item   $\mathtt{O}_\mb{T}$ is given by the following truth-tables:

\begin{table}[h]
\centering
\begin{tabular}{c|c}
$\ov{\neg} (\varphi)$ & $\varphi$ \\
\hline
$\mb{P}$ & $\mb{R}$ \\
$\mb{t}$ & $\mb{f}$ \\
$\mb{f}$ & $\mb{t}$ \\
$\mb{R}$ & $\mb{P}$ \\
\end{tabular}
\quad
\begin{tabular}{c|c}
$\ov{\B}(\varphi)$ & $\varphi$ \\
\hline
$\mb{P}$ & $\mb{P}$ \\
$\mb{\{f,R\}}$ & $\mb{t}$ \\
$\mb{\{f,R\}}$ & $\mb{f}$ \\
$\mb{R}$ & $\mb{R}$ \\
\end{tabular}
\quad
\begin{tabular}{c|cccc}
$\jt$ & $\mb{P}$ & $\mb{t}$ & $\mb{f}$ & $\mb{R}$ \\
\hline
$\mb{P}$  & $\mb{P}$ & $\mb{t}$     & $\mb{f}$    & $\mb{R}$ \\
$\mb{t}$  & $\mb{P}$ & $\mb{\{P,t\}}$ & $\mb{f}$& $\mb{f}$ \\
$\mb{f}$  & $\mb{P}$ & $\mb{\{P,t\}}$ & $\mb{\{P,t\}}$  & $\mb{t}$ \\
$\mb{R}$ & $\mb{P}$ &$\mb{P}$      & $\mb{P}$   & $\mb{P}$ \\
\end{tabular}

\caption{\textbf{T-BAT} truth-tables}
\end{table}

\end{enumerate}
\end{df}

Let's go into more detailed regarding the truth-tables. The negation truth-table maps each value to its natural opposite. When we negate a sentence that is provable $\mb{P}$, we get a refutable falsehood $\mb{R}$. This stems from how the negation of a proven statement is refuted by that very proof. Similarly, negating a refutable falsehood $\mb{R}$ gives us a provable truth $\mb{P}$. For independent statements, negating a true but unprovable statement $\mb{t}$ yields a false but unrefutable one $\mb{f}$ and vice versa. This preserves both the truth-functional behavior of negation and its interaction with provability status.

For implication, let's assume for the sake of the argument that Peano arithmetic, particularly the standard model, is a good approximation of whether something is true or not. Clearly, any implication whose consequent is provable is going to be provable --- by the power of classical logic. Similarly, by the power of classical logic, if the antecedent is provable, the first row of the truth-table of implication is very straightforward. Analogously for the last row. The remaining cases need some examples:
\begin{enumerate}
\item $\mb{t}\jt \mb{t}=\mb{P}$. Simply take twice the sentences $\mathtt{Con}_{\mathbf{PA}}$ asserting the consistency of the Peano arithmetic. It is true and the implication $\mathtt{Con}_{\mathbf{PA}} \jt \mathtt{Con}_{\mathbf{PA}}$ is provable. It is also possible to find less trivial examples.

\item $\mb{t}\jt \mb{t}=\mb{t}$. Take $\mathtt{Bew_{\mathbf{PA}}}(\gn{\mathtt{Con_{\mathbf{PA}}}}) \jt \mathtt{Con}_\mathbf{PA}$.

\item $\mb{t}\jt \mb{f}=\mb{f}$. The implication of any such sentences is false in the standard model, thus it cannot be proven and it cannot be refused, otherwise we would be able to prove unprovable sentences.

\item $\mb{t}\jt \mb{R}=\mb{f}$. Same justification as in the previous example.

\item $\mb{f}\jt \mb{t}=\mb{P}$. The example here is quite convoluted. We construct two sentences, $\varphi$ and $\psi$, such that:
\begin{itemize}
    \item $\varphi$ asserts that for every proof of $\varphi$, there exists a smaller proof of $\psi$,
    \item $\psi$ asserts that for every proof of $\psi$, there exists a smaller proof of $\varphi$.
\end{itemize}
First, both $\varphi$ and $\psi$ are true in the standard model of $\mb{PA}$. \footnote{If $\varphi$ were false, there would be a standard proof of $\varphi$ without a smaller proof of $\psi$, making $\varphi$ a provable but false statement, contradicting the fact that the standard model is the model of \textbf{PA}. A similar argument applies to $\psi$.}.

Second, neither $\varphi$ nor $\psi$ is provable in $\mathbf{PA}$ but $\n\varphi \jt \psi$ is provable. \footnote{For the details of the proofs check \url{https://math.stackexchange.com/questions/peano-arithmetic/63259\#63259}.}

\item $\mb{f}\jt \mb{t}=\mb{t}$.  Consider $\n \mathtt{Con}_{\mathbf{PA}} \jt \mathtt{Con}_{\mathbf{PA}}$.

\item $\mb{f}\jt \mb{f}=\mb{P}$. Consider $\n \mathtt{Con}_{\mathbf{PA}} \jt \n \mathtt{Con}_{\mathbf{PA}}$.
\item $\mb{f}\jt \mb{f}=\mb{t}$. Take any $\vp$ that is false and independent. Then the implication $\mathtt{Bew_{\mathbf{PA}}}(\gn{\vp}) \jt\vp$ is true but not provable.

\item $\mb{f}\jt \mb{R}=\mb{t}$. Stems directly from the power of classical logic.

\end{enumerate}

\subsection{Logics in the vicinity of \textbf{T-BAT}}

The purpose of this section is to explore the range of logics that the suggested semantics can define. Our focus will not be on \textbf{T-BAT} logic. Instead, our analysis will revolve around the simplest logic definable in this framework. This logic stands as the minimal logic among the ones that we will consider in this paper. Additionally, one can see this logic as the minimal modal logic validating axiom \textsf{T}. The logic is induced by the following nmatrix:

\begin{df}[logic \textbf{W}]

$\mathbf{W}= (\mathtt{Val}=\{\mb{P,t,f,R}\},\mathtt{D}=\{\mb{P,t}\},\mathtt{O})$, where $\mathtt{O}$ is given by:

\begin{table}[h!]
\begin{multicols}{3}

  \begin{center}
  \begin{tabular}{l|l}

    $\n$   & $\varphi$  \\ \midrule
    $\ov{\mb{D}}$ & $\mb{P}$  \\ \hline
    $\ov{\mb{D}}$ & $\mb{t}$  \\ \hline
    $\mb{D}$ & $\mb{f}$  \\ \hline
    $\mb{D}$ & $\mb{R}$  \\

  \end{tabular}
\end{center}

\columnbreak

  \begin{tabular}{l|l}
  
    $\B\varphi$   & $\varphi$  \\ \midrule
    $\mb{D}$      & $\mb{P}$  \\ \hline
    $\ov{\mb{D}}$ & $\mb{t}$  \\ \hline
    $\ov{\mb{D}}$ & $\mb{f}$  \\ \hline
    $\ov{\mb{D}}$ & $\mb{R}$  \\

  \end{tabular}

\columnbreak

  \begin{tabular}{l|llll}
   
    $\jt$    & $\mb{P}$ & $\mb{t}$ & $\mb{f}$ & $\mb{R}$ \\ \midrule
    $\mb{P}$ & $\mb{D}$ & $\mb{D}$ & $\ov{\mb{D}}$ & $\ov{\mb{D}}$  \\ \hline
    $\mb{t}$ & $\mb{D}$ & $\mb{D}$ & $\ov{\mb{D}}$ & $\ov{\mb{D}}$   \\ \hline
    $\mb{f}$ & $\mb{D}$ & $\mb{D}$ & $\mb{D}$ & $\mb{D}$   \\ \hline
    $\mb{R}$ & $\mb{D}$ & $\mb{D}$ & $\mb{D}$ & $\mb{D}$  \\
    
  \end{tabular}
\end{multicols}
\caption{Modal logic \textbf{W}}
\end{table}
\end{df}

From the point of proof theory, it is characterized as the minimal extension of the axiomatization of propositional logic plus axiom \textsf{T}:

\begin{description}[itemsep=0pt, leftmargin=2em, style=sameline, font={\mdseries\rmfamily}]
      \item[\textsf{p1}] $\vp \jt (\psi \jt \vp)$.
      \item[\textsf{p2}] $(\vp \jt (\psi \jt \chi))  \jt  (   (\vp \jt \psi)    \jt (\vp \jt \chi)  )$.
      \item[\textsf{p3}] $(\n \psi \jt \n \vp) \jt (\vp \jt \psi)$.
      \item[\textsf{T}] $\B\vp \jt \vp$.
      \item[\textsf{MP}] Rule of \emph{modus ponens}.
      \end{description}

We will use $\Gamma\vdash_{\mathbf{W}} \vp$ to denote its syntactic consequence relation and $\Gamma\vDash_{\mathbf{W}} \vp$ for its semantic consequence relation. Note that axioms \textsf{p1-p3} together with \emph{modus ponens} form an axiomatization of the classical propositional logic. We proceed with the proof of soundness.
\begin{thm}
For any $\Gamma,\vp$ if $\Gamma\vdash_{\mathbf{W}} \vp$, then $\Gamma\vDash_{\mathbf{W}} \vp$.
\end{thm}
\begin{proof}
Straightforward induction on the complexity of $\vp$.
\end{proof}

Prior to demonstrating completeness, we will begin by revisiting the standard concepts typically employed in proving completeness.

\begin{df}
Let $\Gamma\subseteq \mc{FR}$, then $\Gamma$ is an $\mb{L}$-\emph{consistent} set iff $\Gamma\not\vdash_{\mb{L}} \vp$ or $\Gamma\not\vdash_{\mb{L}} \neg \vp$ for all $\vp\in \mc{FR}$. $\Gamma$ is $\mb{L}$-\emph{inconsistent} otherwise.
\end{df}

\begin{df}
Let $\Gamma\subseteq \mc{FR}$, then $\Gamma$ is \emph{maximal} $\mb{L}$-\emph{consistent} set iff $\Gamma$ is $\mb{L}$-\emph{consistent} and any set of formulas properly containing $\Gamma$ is $\mb{L}$-inconsistent. If $\Gamma$ is maximal $\mb{L}$-consistent set, then we say that $\Gamma$ is an $\mb{L}$-$\mathtt{mxc}$.
\end{df}

The following well-known lemma is valid for any Tarskian and finitary logic. In particular, it is valid for all logics considered in this paper.

\begin{fct}
For any $\Gamma\cup \{ \vp \}\subseteq \mc{FR}$, if $\Gamma\not\vdash_{\mb{L}} \vp$, then, there is a $\Gamma^{\mt{mxc}} \supseteq \Gamma$ such that $\Gamma^{\mt{mxc}}$ is $\mb{L}$-$\mathtt{mxc}$ and $\Gamma^{\mt{mxc}}\not\vdash_{\mb{L}}\vp$.
\end{fct}

The proof of completeness is based on one key lemma:\footnote{We will follow the strategy suggested by \cite{Omori2016,Pawlowski2022a} }

\begin{lm}[Valuation lemma]

If $\Gamma$ is $\mb{W}$-$\mathtt{mxc}$, then  the function $\Theta_\Gamma: \mc{FR} \mapsto \mathtt{Val}$ defined in the following way:
\[ 
\Theta_\Gamma(\vp) = 
\begin{cases}
\mathbf{P} & \text{ iff } \Gamma \vdash_{\mathbf{W}} \varphi \text{ and } \Gamma \vdash_{\mathbf{W}} \Box \varphi  \\
\mathbf{t} & \text{ iff } \Gamma \vdash_{\mathbf{W}} \varphi \text{ and } \Gamma \vdash_{\mathbf{W}} \n\Box \varphi   \\
\mathbf{f} & \text{ iff } \Gamma \vdash_{\mathbf{W}} \n\varphi \text{ and } \Gamma \vdash_{\mathbf{W}} \n\Box \n \varphi   \\
\mathbf{R} & \text{ iff } \Gamma \vdash_{\mathbf{W}} \n\varphi \text{ and } \Gamma \vdash_{\mathbf{W}} \B\n\varphi  \\
\end{cases}
\]
is well-defined $\mathbf{W}$-valuation. 
\end{lm}
\begin{proof}
We prove it by induction on the complexity of $\lambda$. The base case is straightforward. Assume the theorem holds for $\vp, \psi$. We will prove it for $\B\vp,\n\vp$ and $\vp\jt\psi$. 

\begin{enumerate}
\item $\lambda:= \n\vp$. We start with $\Theta_\Gamma(\vp)\in \mt{D}$. By definition this means that $\Gamma\vdash_{\mathsf{W}} \vp$. By classical propositional logic we get $\Gamma\vdash_{\mathsf{W}} \n\n\vp$. So, we get $\Theta_\Gamma(\n\vp)\in \ov{\mt{D}}$. If  $\Theta_\Gamma(\vp)\in \ov{\mt{D}}$, then we have  $\Gamma\vdash_{\mathsf{W}} \n\vp$ and so $\Theta_\Gamma(\n\vp)\in \ov{\mt{D}}$.
\item $\lambda:= \vp\jt \psi$. If $\Theta_\Gamma(\psi)\in \mt{D}$ by definition we have  $\Gamma\vdash_{\mathsf{W}} \psi$ and so by classical propositional logic $\Gamma\vdash_{\mathsf{W}} \vp \jt\psi$ giving us  $\Theta_\Gamma(\vp\jt \psi)\in \mt{D}$. If $\Theta_\Gamma(\psi)\in \ov{\mt{D}}$ and $\Theta_\Gamma(\vp)\in \mt{D}$, we have $\Gamma\vdash_{\mathsf{W}} \n\psi$ and $\Gamma\vdash_{\mathsf{W}} \vp$ giving us $\Gamma\vdash_{\mathsf{W}}\n(\vp \jt \psi)$, ultimately resulting in $\Theta_\Gamma(\vp\jt \psi)\in \ov{\mt{D}}$. If $\Theta_\Gamma(\psi)\in \ov{\mt{D}}$ and $\Theta_\Gamma(\vp)\in \ov{\mt{D}}$ we have $\Gamma\vdash_{\mathsf{W}} \n \vp$, $\Gamma\vdash_{\mathsf{W}} \n\psi$ and so $\Gamma\vdash_{\mathsf{W}} \vp \jt \psi$, giving  $\Theta_\Gamma(\vp\jt \psi)\in \mt{D}$.
\item $\lambda:= \B\vp$. We start with $\Theta_\Gamma(\vp)\in \{\mt{t,f}\}$. So $\Gamma\vdash_{\mathsf{W}} \n\B\vp$ which gives us $\Theta_\Gamma(\B\vp)\in \ov{\mt{D}}$. If  $\Theta_\Gamma(\vp)=\mt{R}$, we have $\Gamma\vdash_{\mathsf{W}} \B\n\vp$ and by axiom $\mathsf{T}$ and classical propositional logic, $\Gamma\vdash_{\mathsf{W}} \n\B\vp$, so $\Theta_\Gamma(\B\vp)\in \ov{\mt{D}}$. The last case is $\Theta_\Gamma(\vp)=\mt{P}$. We have $\Gamma\vdash_{\mathsf{W}} \B\vp$, so $\Theta_\Gamma(\B\vp)\in \mt{D}$.

\end{enumerate} 

\end{proof}

\begin{thm}
$\Gamma \vDash_{\mathbf{W}}\vp$ iff $\Gamma \vdash_{\mathbf{W}}\vp$, 
\end{thm}
\begin{proof}
The proof is by contraposition and in the light of the valuation lemma is straightforward.
\end{proof}

This logic is going to be our starting point. Any other logic that we will consider in this paper is obtained by strengthening the interpretation of some of the connectives. So, by making connectives more deterministic we enforce the validity of additional axioms. Before we proceed with the study of the extensions let us note two  crucial remarks.

\begin{rmk}
In this particular setting it is possible to simplify the valuation function. In the light of validity of $\B\vp \jt \vp$ and $\B\n\vp \jt \n\vp$ the simplified valuation function is:
\[ 
\Theta^{\mathtt{S}}_\Gamma(\vp) = 
\begin{cases}
\mathbf{P} & \text{ iff } \Gamma \vdash_{\mathbf{W}} \Box \varphi  \\
\mathbf{t} & \text{ iff } \Gamma \vdash_{\mathbf{W}} \varphi \text{ and } \Gamma \vdash_{\mathbf{W}} \n\Box \varphi   \\
\mathbf{f} & \text{ iff } \Gamma \vdash_{\mathbf{W}} \n\varphi \text{ and } \Gamma \vdash_{\mathbf{W}} \n\Box \n \varphi   \\
\mathbf{R} & \text{ iff }  \Gamma \vdash_{\mathbf{W}} \B\n\varphi  \\
\end{cases}
\]
\end{rmk}

\begin{rmk}
Importantly, the presented axiomatization is also sound and complete with respect to the following interpretation of connectives:\\

\begin{table}[h!]
\centering
\begin{multicols}{3}

  \begin{center}
  \begin{tabular}{|l|l|}
    \toprule
    $\n$   & $\varphi$  \\ \midrule
    $\ov{\mb{D}}$ & $\mb{P}$  \\ \hline
    $\ov{\mb{D}}$ & $\mb{t}$  \\ \hline
    $\mb{t}$ & $\mb{f}$  \\ \hline
    $\mt{P}$ & $\mb{R}$  \\
    \bottomrule
  \end{tabular}
\end{center}

\columnbreak

  \begin{tabular}{|l|l|}
    \toprule
    $\B\varphi$   & $\varphi$  \\ \midrule
    $\mb{D}$      & $\mb{P}$  \\ \hline
    $\ov{\mb{D}}$ & $\mb{t}$  \\ \hline
    $\ov{\mb{D}}$ & $\mb{f}$  \\ \hline
    $\ov{\mb{D}}$ & $\mb{R}$  \\
    \bottomrule
  \end{tabular}

\columnbreak

  \begin{tabular}{|l|l|l|l|l|}
    \toprule
    $\jt$    & $\mb{P}$ & $\mb{t}$ & $\mb{f}$ & $\mb{R}$ \\ \midrule
    $\mb{P}$ & $\mb{D}$ & $\mb{D}$ & $\ov{\mb{D}}$ & $\ov{\mb{D}}$  \\ \hline
    $\mb{t}$ & $\mb{D}$ & $\mb{D}$ & $\ov{\mb{D}}$ & $\ov{\mb{D}}$   \\ \hline
    $\mb{f}$ & $\mb{D}$ & $\mb{D}$ & $\mb{D}$ & $\mb{D}$   \\ \hline
    $\mb{R}$ & $\mb{D}$ & $\mb{D}$ & $\mb{D}$ & $\mb{D}$  \\
    \bottomrule
  
  \end{tabular}
  \end{multicols}
  \caption{Simplified semantics}

\end{table}
\noindent To see this consider the appropriate case in the proof  regarding the valuation function i.e. $\lambda:= \n\vp$. We assume that $\Theta_\Gamma(\vp)= \mt{R}$. By definition this means that $\Gamma\vdash_{\mathbf{W}} \B\n\vp$. We need to show that $\Theta_\Gamma(\n\vp)= \mt{P}$. Thus, we need to show that $\Gamma\vdash_{\mathbf{W}} \B\n\vp$ which we do have by definition. For the other case assume that $\Theta_\Gamma(\vp)= \mt{f}$ which means that $\Gamma\vdash_{\mathbf{W}} \n\vp$ and $\Gamma\vdash_{\mathbf{W}} \n\B\n\vp$. We need to show that $\Theta_\Gamma(\n\vp)= \mt{t}$ i.e. we need to check if $\Gamma\vdash_{\mathbf{W}} \n\vp$ and $\Gamma\vdash_{\mathbf{W}} \n\B\n\vp$ which indeed is the case.

\end{rmk}
\noindent Moving forward, we will adopt the strengthened semantics approach. It is important to note that any further enhancement of this semantics will necessitate the introduction of new axioms. Before we dive into the specific axioms, let us outline the methodology we use to derive them. Notice that the meaning of the values used in this semantics is codified by the valuation lemma. We will exploit this fact by correlating each strengthening of the interpretation of a given connective with a formula that is obtained by directly translating the meaning of the values into the object language. To illustrate this more effectively, let us examine the negation and specifically its first row:

\begin{center} 
If $v(\vp)=\mt{P}$, then $v(\n\vp)=\mt{R}$.
\end{center}
\noindent The values are deciphered using their designated meanings from the revised version of the valuation lemma. This means $\mt{P}$, indicating provability and truth, is expressed as $\B\vp$, and $\n\vp$, holding the value $\mt{R}$, is interpreted as $\B\n\n\vp$. These interpretations are then intertwined through an implication, forming the axiom:

\[\B\vp \jt \B\n\n\vp\]

\noindent All the possible variants of the interpretations of negation are summarize in the following table:\footnote{Slightly similar investigations solely focused on $\Box$ has been done in \cite{Pawlowski2022a}.}

\begin{table}[h!]
\begin{center}
\begin{tabular}{|l|l|l|l|l|}
\hline
Name&$\vp$    & $\n\vp$  & Axiom                                         & Simplification \\ \hline
$\mb{N1}$ & $\mt{P}$ & $\mt{R}$ & $\B\vp \jt \B\n\n\vp $                         & None \\ \hline
$\mb{N2}$ & $\mt{P}$ & $\mt{f}$ & $\B\vp  \jt \n\B\n\n\vp \et \n\n\vp $          & $\B\vp \jt \n\B\n\n\vp $    \\ \hline
$\mb{N3}$ & $\mt{t}$ & $\mt{R}$ & $\n\B\vp \et \vp \jt \B\n\n\vp$                &  None    \\ \hline
$\mb{N4}$ & $\mt{t}$ & $\mt{f}$ & $\n\B\vp \et \vp \jt \n\B\n\n\vp \et \n\n\vp$  & $\n\B\vp \et \vp \jt \n\B\n\n\vp $      \\ \hline
 \end{tabular}
 \end{center}
\caption{Various strengthenings of negation}
\end{table}
\noindent While certain extensions fail to yield profound interpretations of $\B$ as a modality, others, like axioms $\mb{N1,N4}$, are provable in modal logic \textbf{K}. This is attributed to the validity  of the substitution of provably equivalent formulas within the scope of a modal operator. This principle is valid in Kripke semantics and neighborhood semantics. In contrast, non-deterministic semantics do not inherently validate this principle so the semantics can make a distinction between provably equivalent formulas. This distinction allows non-deterministic semantics to serve as a specialized instrument for investigating hyper-intensionality.\footnote{Worth noticing that this perspective has yet to be developed.}\footnote{All the calculations regarding the Kripke frame conditions have been done using \url{https://store.fmi.uni-sofia.bg/fmi/logic/sqema/sqema_gwt_20180317_2/index.html}.}

Axiom $\mb{N3}$ is  slightly more interesting. First, it is not a theorem of \textbf{S5} but it seems to have some modal meaning. Frame condition-wise it is equivalent to the following:
\[ \forall_{x,y} \, xRy \jt x=y \]

The last axiom is $\mb{N2}$ for which we were unable to track down the corresponding frame condition.

We are going to proceed with this methodology and apply it to the case of $\B$. This generate the following table:

\begin{table}[h!]
\begin{center}
\begin{tabular}{|l|l|l|l|l|}
\hline
Name      & $\vp$    & $\B\vp$  & Axiom                                            & Simplification \\ \hline
$\mb{B1}$ & $\mt{P}$ & $\mt{P}$ & $\B\vp \jt \B\B\vp $                             & None \\ \hline
$\mb{B2}$ & $\mt{P}$ & $\mt{t}$ & $\B\vp  \jt \B\vp \et \n\B\B\vp $                & $\B\vp \jt \n\B\B\vp$ \\ \hline
$\mb{B3}$ & $\mt{t}$ & $\mt{R}$ & $\n\B\vp \et \vp \jt \B\n\B\vp$                  & None    \\ \hline
$\mb{B4}$ & $\mt{t}$ & $\mt{f}$ & $\n\B\vp \et \vp \jt \n\B\vp \et \n\B\n\B\vp$      & $\n\B\vp \et \vp \jt \n\B\n\B\vp$   \\ \hline
$\mb{B5}$ & $\mt{f}$ & $\mt{R}$ & $\n\B\n\vp \et \n\vp \jt \B\n\B\vp$              & None   \\ \hline
$\mb{B6}$ & $\mt{f}$ & $\mt{f}$ & $\n\B\n\vp \et \n\vp \jt \n\B\vp \et \n\B\n\B\vp $ & None       \\ \hline
$\mb{B7}$ & $\mt{R}$ & $\mt{R}$ & $\B\n\vp \jt \B\n\B\vp$                        & None   \\ \hline
$\mb{B8}$ & $\mt{R}$ & $\mt{f}$ & $\B\n\vp \jt \n\B\vp \et \n\B\n\B\vp$            & None   \\ \hline
 \end{tabular}
 \end{center}
\caption{Various strengthenings of $\B$}
\end{table}

\noindent $\mb{B1}$ is a well-known modal axiom, axiom \textbf{4}. Frame wise this axiom is equivalent to reflexivity.

Axiom $\mb{B2}$ modally speaking is not interesting and seems not to correspond to any interesting frame condition. $\mb{B3}$ is a theorem of \textbf{S4} and induces the following condition on frames:
\[\forall_{x,y,z} \,(xRy \jt (xRz \jt (x = z \vee  yRz)))\]

Axiom $\mb{B4}$ is responsible for the following:

\[\forall_{x,y,z,z_1} \, (xRy \jt x = y) \vee  \exists z(xRz \land (zRz_1 \jt x = z_1))\]

Axiom $\mb{B5}$ is provable in the modal logic \textbf{B} and is equivalent to the following frame-condition:

\[ \forall_{x,y,z} \,((xRy \land  xRz) \jt (x =z  \vee  yRx))\]

Axiom $\mb{B6}$ is not provable in \textbf{S5} and is equivalent to the following condition:
\[\forall_{x,y_1,y_2}\, (xRy_1\jt  (x = y_1 \vee  \exists_{z_1} \,(xRz_1 \land  \forall_{z_2} \, (z_1Rz_2 \jt y_1 = z_2)))) \land  (xRy_2 \jt (x = y_2 \vee  xRx))
 \]

Axiom $\mb{B7}$ is provable in \textbf{T} and modally gives:

\[\forall_{x,y} \,(xRy  \jt  \exists_{z} \, (xRz \land  yRz))\]
and $\mb{B8}$ is not provable in \textsf{S5} and results in the following condition:

\[\forall_{x}\exists_{y}\, xRy \land  \exists_{z}\, z(xRz \land  \forall_{z_1} \,\neg zRz_1)\]

The last part of this study has to do with various strengthenings of the conditional. This is the most tricky part, since it involves the most combinations. We start with the following labeling, together with the axiom and some additional info about a given axiom.

\begin{table}[h!]
\tiny
\begin{center}
\begin{tabular}{|l|l|l|l|l|l|}
\hline
Name                           & $ \vp $    & $ \psi  $   & $ \vp  \jt  \psi  $ & Axiom                                                                              & Remark \\ \hline
$\mb{I}_{{\mt{P,P}}}^{\mt{P}}$ & $\mt{P}$ & $\mt{P}$ & $\mt{P}$ & $\B \vp  \et \B  \psi   \jt \B( \vp  \jt  \psi  )$                                                 & Provable in \textbf{K} \\ \hline
$\mb{I}_{{\mt{P,P}}}^{\mt{t}}$ & $\mt{P}$ & $\mt{P}$ & $\mt{t}$ & $\B \vp  \et \B  \psi  \jt \n\B( \vp  \jt  \psi  ) \et ( \vp  \jt  \psi  )$                              & Not provable in \textbf{S5} \\ \hline
$\mb{I}_{{\mt{P,t}}}^{\mt{P}}$ & $\mt{P}$ & $\mt{t}$ & $\mt{P}$ & $\B \vp   \et \n\B  \psi   \et  \psi   \jt \B( \vp  \jt  \psi  )$                                      & Not provable in \textbf{S5} \\ \hline
$\mb{I}_{{\mt{P,t}}}^{\mt{t}}$ & $\mt{P}$ & $\mt{t}$ & $\mt{t}$ & $\B \vp   \et \n\B  \psi   \et  \psi   \jt \n\B( \vp  \jt  \psi  ) \et ( \vp  \jt  \psi  )$                  & Provable in \textbf{K}\\ \hline
$\mb{I}_{{\mt{P,f}}}^{\mt{f}}$ & $\mt{P}$ & $\mt{f}$ & $\mt{f}$ & $\B \vp   \et \n\B\n  \psi   \et \n \psi   \jt \n\B\n( \vp  \jt  \psi  ) \et \n( \vp  \jt  \psi  )$          & Provable in \textbf{T} \\ \hline
$\mb{I}_{{\mt{P,f}}}^{\mt{R}}$ & $\mt{P}$ & $\mt{f}$ & $\mt{R}$ & $\B \vp   \et \n\B\n  \psi   \et \n \psi   \jt \B\n( \vp  \jt  \psi  )$                                & Not provable in \textbf{S5} \\ \hline
$\mb{I}_{{\mt{P,R}}}^{\mt{f}}$ & $\mt{P}$ & $\mt{R}$ & $\mt{f}$ & $\B \vp   \et \B\n  \psi    \jt \n\B\n( \vp  \jt  \psi  ) \et \n( \vp  \jt  \psi  )$                     & Not provable in \textbf{S5}  \\ \hline
$\mb{I}_{{\mt{P,R}}}^{\mt{R}}$ & $\mt{P}$ & $\mt{R}$ & $\mt{R}$ & $\B \vp   \et \B\n  \psi    \jt \B\n( \vp  \jt  \psi  ) $                                          & Provable in \textbf{K} \\ \hline
$\mb{I}_{{\mt{t,P}}}^{\mt{P}}$ & $\mt{t}$ & $\mt{P}$ & $\mt{P}$ & $\n\B \vp  \et  \vp  \et \B  \psi    \jt \B( \vp  \jt  \psi  ) $                                     & Provable in \textbf{K}  \\ \hline
$\mb{I}_{{\mt{t,P}}}^{\mt{t}}$ & $\mt{t}$ & $\mt{P}$ & $\mt{t}$ & $\n\B \vp  \et  \vp  \et \B  \psi    \jt \n\B( \vp  \jt  \psi  ) \et ( \vp  \jt  \psi  )$                  & Not provable in \textbf{S5} \\ \hline
$\mb{I}_{{\mt{t,t}}}^{\mt{P}}$ & $\mt{t}$ & $\mt{t}$ & $\mt{P}$ & $\n\B \vp  \et  \vp  \et \n\B\  \psi   \et  \psi   \jt \B( \vp  \jt  \psi  ) $                           & Not provable in \textbf{S5} \\ \hline
$\mb{I}_{{\mt{t,t}}}^{\mt{t}}$ & $\mt{t}$ & $\mt{t}$ & $\mt{t}$ & $\n\B \vp  \et  \vp  \et \n\B\  \psi   \et  \psi   \jt \n\B( \vp  \jt  \psi  ) \et ( \vp  \jt  \psi  )$        & Provable in \textbf{K}  \\ \hline
$\mb{I}_{{\mt{t,f}}}^{\mt{f}}$ & $\mt{t}$ & $\mt{f}$ & $\mt{f}$ & $\n\B \vp  \et  \vp  \et \n\B\n  \psi   \et \n \psi   \jt \n\B\n( \vp  \jt  \psi  ) \et \n( \vp  \jt  \psi  )$ & Provable in \textbf{K} \\ \hline
$\mb{I}_{{\mt{t,f}}}^{\mt{R}}$ & $\mt{t}$ & $\mt{f}$ & $\mt{R}$ & $\n\B \vp  \et  \vp  \et \n\B\n  \psi   \et \n \psi   \jt \B\n( \vp  \jt  \psi  )$                       & Not provable in \textbf{S5}\\ \hline
$\mb{I}_{{\mt{t,R}}}^{\mt{f}}$ & $\mt{t}$ & $\mt{R}$ & $\mt{f}$ & $\n\B \vp  \et  \vp  \et \B\n  \psi    \jt \n\B\n( \vp  \jt  \psi  ) \et \n( \vp  \jt  \psi  )$            & Provable in \textbf{T} \\ \hline
$\mb{I}_{{\mt{t,R}}}^{\mt{R}}$ & $\mt{t}$ & $\mt{R}$ & $\mt{R}$ & $\n\B \vp  \et  \vp  \et \B\n  \psi    \jt \B\n( \vp  \jt  \psi  )$              & Not provable in \textbf{S5} \\ \hline
$\mb{I}_{{\mt{f,P}}}^{\mt{P}}$ & $\mt{f}$ & $\mt{P}$ & $\mt{P}$ & $\n\B\n \vp  \et \n \vp  \et \B  \psi    \jt \B( \vp  \jt  \psi  ) $                                 & Provable in \textbf{K}  \\ \hline
$\mb{I}_{{\mt{f,P}}}^{\mt{t}}$ & $\mt{f}$ & $\mt{P}$ & $\mt{t}$ & $\n\B\n \vp  \et \n \vp  \et \B  \psi    \jt \n\B( \vp  \jt  \psi  ) \et ( \vp  \jt  \psi  )$              & Not provable in \textbf{S5} \\ \hline
$\mb{I}_{{\mt{f,t}}}^{\mt{P}}$ & $\mt{f}$ & $\mt{t}$ & $\mt{P}$ & $\n\B\n \vp  \et \n \vp  \et \n\B  \psi   \et  \psi   \jt \B( \vp  \jt  \psi  ) $                        & Not provable in \textbf{S5}\\ \hline
$\mb{I}_{{\mt{f,t}}}^{\mt{t}}$ & $\mt{f}$ & $\mt{t}$ & $\mt{t}$ & $\n\B\n \vp  \et \n \vp  \et \n\B  \psi   \et  \psi   \jt \n\B( \vp  \jt  \psi  ) \et ( \vp  \jt  \psi  )$     & Not provable in \textbf{S5}\\ \hline
$\mb{I}_{{\mt{f,f}}}^{\mt{P}}$ & $\mt{f}$ & $\mt{f}$ & $\mt{P}$ & $\n\B\n \vp  \et \n \vp  \et \n\B\n  \psi   \et \n \psi   \jt \B( \vp  \jt  \psi  ) $                    & Not provable in \textbf{S5} \\ \hline
$\mb{I}_{{\mt{f,f}}}^{\mt{t}}$ & $\mt{f}$ & $\mt{f}$ & $\mt{t}$ & $\n\B\n \vp  \et \n \vp  \et \n\B\n  \psi   \et \n \psi   \jt \n\B( \vp  \jt  \psi  ) \et ( \vp  \jt  \psi  )$ & Not provable in \textbf{S5}\\ \hline
$\mb{I}_{{\mt{f,R}}}^{\mt{P}}$ & $\mt{f}$ & $\mt{R}$ & $\mt{P}$ & $\n\B\n \vp  \et \n \vp  \et \B\n  \psi    \jt \B( \vp  \jt  \psi  ) $                               & Not provable in \textbf{S5} \\ \hline
$\mb{I}_{{\mt{f,R}}}^{\mt{t}}$ & $\mt{f}$ & $\mt{R}$ & $\mt{t}$ & $\n\B\n \vp  \et \n \vp  \et \B\n  \psi    \jt \n\B( \vp  \jt  \psi  ) \et ( \vp  \jt  \psi  )$            & Provable in \textbf{K}    \\ \hline
$\mb{I}_{{\mt{R,P}}}^{\mt{P}}$ & $\mt{R}$ & $\mt{P}$ & $\mt{P}$ & $\B\n \vp  \et \B  \psi    \jt \B( \vp  \jt  \psi  ) $                                             & Provable in \textbf{K}   \\ \hline
$\mb{I}_{{\mt{R,P}}}^{\mt{t}}$ & $\mt{R}$ & $\mt{P}$ & $\mt{t}$ & $\B\n \vp  \et \B  \psi    \jt \n\B( \vp  \jt  \psi  ) \et ( \vp  \jt  \psi  )$                          & Not provable in \textbf{S5} \\ \hline
$\mb{I}_{{\mt{R,t}}}^{\mt{P}}$ & $\mt{R}$ & $\mt{t}$ & $\mt{P}$ & $\B\n \vp  \et \n\B  \psi   \et  \psi   \jt \B( \vp  \jt  \psi  ) $                                    & Provable in \textbf{K}    \\ \hline
$\mb{I}_{{\mt{R,t}}}^{\mt{t}}$ & $\mt{R}$ & $\mt{t}$ & $\mt{t}$ & $\B\n \vp  \et \n\B  \psi   \et  \psi   \jt \n\B( \vp  \jt  \psi  ) \et ( \vp  \jt  \psi  )$                 & Not provable in \textbf{S5}\\ \hline
$\mb{I}_{{\mt{R,f}}}^{\mt{P}}$ & $\mt{R}$ & $\mt{f}$ & $\mt{P}$ & $\B\n \vp  \et \n\B\n  \psi   \et \n \psi   \jt \B( \vp  \jt  \psi  )$                                 & Provable in \textbf{K}    \\ \hline
$\mb{I}_{{\mt{R,f}}}^{\mt{t}}$ & $\mt{R}$ & $\mt{f}$ & $\mt{t}$ & $\B\n \vp  \et \n\B\n  \psi   \et \n \psi   \jt \n\B( \vp  \jt  \psi  ) \et ( \vp  \jt  \psi  )$             & Not provable in \textbf{S5} \\ \hline
$\mb{I}_{{\mt{R,R}}}^{\mt{P}}$ & $\mt{R}$ & $\mt{R}$ & $\mt{P}$ & $\B\n \vp  \et \B\n  \psi    \jt \B( \vp  \jt  \psi  )$                                            & Provable in \textbf{K}  \\ \hline
$\mb{I}_{{\mt{R,R}}}^{\mt{t}}$ & $\mt{R}$ & $\mt{R}$ & $\mt{t}$ & $\B\n \vp  \et \B\n  \psi    \jt \n\B( \vp  \jt  \psi  ) \et ( \vp  \jt  \psi  )$                        & Not provable in \textbf{S5}\\ \hline
\end{tabular}
\end{center}
\caption{Various implications}
\end{table}
\noindent Based on the discussion, we can categorize the axioms into three distinct groups. The first group includes axioms that are typically provable with the presence of either axiom \textsf{K}, the rule of necessitation, or both. In a standard modal framework, these axioms do not offer any novel insights into modality. The second group is notably smaller, consisting of just two axioms. These are not provable in \textbf{K} but can be proven in \textbf{T}. The final group comprises axioms that are not provable in \textbf{S5}.

\begin{table}[h!]
\tiny
\label{Table1}
\begin{center}
\begin{tabular}{|l|l|l|}
\hline Name                            & Axiom                                                                         &Condition                              \\ \hline

$\mb{I}_{{\mt{P,t}}}^{\mt{P}}$ & $\B \vp   \et \n\B  \psi   \et  \psi   \jt \B( \vp  \jt  \psi  )$                                     & $\forall_{x,y} \,(xRy \jt x = y)$      \\ \hline
$\mb{I}_{{\mt{P,f}}}^{\mt{R}}$ & $\B \vp   \et \n\B\n  \psi   \et \n \psi   \jt \B\n( \vp  \jt  \psi  )$                               & $\forall_{x,y} \,(xRy \jt x = y)$                                                 \\ \hline
$\mb{I}_{{\mt{R,t}}}^{\mt{t}}$ & $\B\n \vp  \et \n\B  \psi   \et  \psi   \jt \n\B( \vp  \jt  \psi  ) \et ( \vp  \jt  \psi  )$                & $\forall_{x,y} \,(xRy \jt x = y)$    \\ \hline
$\mb{I}_{{\mt{f,P}}}^{\mt{t}}$ & $\n\B\n \vp  \et \n \vp  \et \B  \psi    \jt \n\B( \vp  \jt  \psi  ) \et ( \vp  \jt  \psi  )$             & $\forall_{x,y} \,(xRy \jt x = y)$ \\ \hline
$\mb{I}_{{\mt{f,R}}}^{\mt{P}}$ & $\n\B\n \vp  \et \n \vp  \et \B\n  \psi    \jt \B( \vp  \jt  \psi  ) $                              & $\forall_{x,y} \,(xRy \jt x = y)$ \\ \hline
$\mb{I}_{{\mt{f,R}}}^{\mt{t}}$ & $\n\B\n \vp  \et \n \vp  \et \B\n  \psi    \jt \n\B( \vp  \jt  \psi  ) \et ( \vp  \jt  \psi  )$           & $\forall_{x,y} \,(xRy \jt x = y \vee  xRy)$ \\ \hline
$\mb{I}_{{\mt{t,f}}}^{\mt{R}}$ & $\n\B \vp  \et  \vp  \et \n\B\n  \psi   \et \n \psi   \jt \B\n( \vp  \jt  \psi  )$                      & $\forall_{x,y,z} \,(xRy \jt (x = y \vee  xRz \jt x = z))$ \\ \hline
$\mb{I}_{{\mt{R,f}}}^{\mt{t}}$ & $\B\n \vp  \et \n\B\n  \psi   \et \n \psi   \jt \n\B( \vp  \jt  \psi  ) \et ( \vp  \jt  \psi  )$            & $\forall_{x,y,z} \,((xRy \jt x = y) \land  (xRz \jt (x = z \vee  xRx)))$      \\ \hline
$\mb{I}_{{\mt{P,f}}}^{\mt{f}}$ & $\B \vp   \et \n\B\n  \psi   \et \n \psi   \jt \n\B\n( \vp  \jt  \psi  ) \et \n( \vp  \jt  \psi  )$         & $\forall_{x,y,z} \,((xRy \jt (x = y \vee  xRy) \land  (xRz \jt (x = z \vee xRx))))$  \\ \hline
$\mb{I}_{{\mt{R,t}}}^{\mt{P}}$ & $\B\n \vp  \et \n\B  \psi   \et  \psi   \jt \B( \vp  \jt  \psi  ) $                                   & $\forall_{x,y,z} \,(xRy \jt (x = y \vee  xRy \vee (xRz \jt x = z)))$    \\ \hline
$\mb{I}_{{\mt{t,t}}}^{\mt{t}}$ & $\n\B \vp  \et  \vp  \et \n\B\  \psi   \et  \psi   \jt \n\B( \vp  \jt  \psi  ) \et ( \vp  \jt  \psi  )$       & $\forall_{x,y,z} \,((xRy \jt x = y) \vee  (xRz \jt (x = z \vee  (x = z \land  xRz))))$\\ \hline
$\mb{I}_{{\mt{t,f}}}^{\mt{f}}$ & $\n\B \vp  \et  \vp  \et \n\B\n  \psi   \et \n \psi   \jt \n\B\n( \vp  \jt  \psi  ) \et \n( \vp  \jt  \psi  )$& $\forall_{x,y,z} \,(xRy \jt (x = y \vee xRy \vee  (xRz \jt (x = z\vee  xRz))))$ \\  \hline
$\mb{I}_{{\mt{f,t}}}^{\mt{t}}$ & $\n\B\n \vp  \et \n \vp  \et \n\B  \psi   \et  \psi   \jt \n\B( \vp  \jt  \psi  ) \et ( \vp  \jt  \psi  )$    & $\forall_{x,y,z} \,(xRy \jt (x = y \vee  (xRz \jt (x = z \vee  (y= z \land  xRz)))))$\\ \hline
$\mb{I}_{{\mt{R,f}}}^{\mt{P}}$ & $\B\n \vp  \et \n\B\n  \psi   \et \n \psi   \jt \B( \vp  \jt  \psi  )$                                & $\forall_{x,y,z} \,(xRy \land  xRz \jt (x = z \vee  y =z  \vee  xRy)$   \\ \hline
$\mb{I}_{{\mt{f,f}}}^{\mt{t}}$ & $\n\B\n \vp  \et \n \vp  \et \n\B\n  \psi   \et \n \psi   \jt \n\B( \vp  \jt  \psi  ) \et ( \vp  \jt  \psi  )$& $\forall_{x,y,z} \,(xRy \land  xRz \jt (x = y \vee  x = z \vee  (x = y \land  xRx)))$\\ \hline
$\mb{I}_{{\mt{t,P}}}^{\mt{t}}$ & $\n\B \vp  \et  \vp  \et \B  \psi    \jt \n\B( \vp  \jt  \psi  ) \et ( \vp  \jt  \psi  )$                 & $\forall_{x,y,z} \,((xRx \vee  (xRy \jt x = y)) \land (xRz \jt x =z ))$ \\ \hline
$\mb{I}_{{\mt{t,R}}}^{\mt{f}}$ & $\n\B \vp  \et  \vp  \et \B\n  \psi    \jt \n\B\n( \vp  \jt  \psi  ) \et \n( \vp  \jt  \psi  )$           & $\forall_{x,y,z} \,((xRx \vee (xRy \jt x = y)) \land (xRz \jt (x = z \vee  xRz)))$ \\ \hline
$\mb{I}_{{\mt{t,R}}}^{\mt{R}}$ & $\n\B \vp  \et  \vp  \et \B\n  \psi    \jt \B\n( \vp  \jt  \psi  ) \et \n( \vp  \jt  \psi  )$             & $\forall_{x,y,z} \,((xRx \vee  (xRy \jt x = y)) \land  (xRz \jt x = z))$ \\ \hline
$\mb{I}_{{\mt{f,P}}}^{\mt{P}}$ & $\n\B\n \vp  \et \n \vp  \et \B  \psi    \jt \B( \vp  \jt  \psi  ) $                                & $\forall_{x,y,z} \,((xRy \land  xRz) \jt (x = y  \vee  x = z  \vee  xRz))$  \\ \hline
$\mb{I}_{{\mt{f,f}}}^{\mt{P}}$ & $\n\B\n \vp  \et \n \vp  \et \n\B\n  \psi   \et \n \psi   \jt \B( \vp  \jt  \psi  ) $                   & $\forall_{x,y_1,y_2,z}\,(xRy_1 \land xRy_2 \land  xRz \jt (x = y_1 \vee  x = y_2 \vee  x =z  \vee  y_2 = z)$ \\ \hline
$\mb{I}_{{\mt{f,t}}}^{\mt{P}}$ & $\n\B\n \vp  \et \n \vp  \et \n\B  \psi   \et  \psi   \jt \B( \vp  \jt  \psi  ) $                       & $\forall_{x,y_1,y_2,z}\,(xRy_1 \land  xRy_2 \jt (x = y_1 \vee  x = y_2 \vee (xRz \jt x = z)))$\\ \hline
$\mb{I}_{{\mt{t,t}}}^{\mt{P}}$ & $\n\B \vp  \et  \vp  \et \n\B\  \psi   \et  \psi   \jt \B( \vp  \jt  \psi  ) $                          & $\forall_{x,y_1,y_2,z}\,(xRy_1 \jt (x = y_1 \vee (xRy_2 \jt x =y_2 ) \vee (xRz \jt (x = z \vee  y_1 = z))))$ \\ \hline

\end{tabular}
\end{center}
\caption{Summary the Kripke conditions}
\end{table}

\noindent However, for four axioms we were unable to find the corresponding conditions:

\begin{table}[h!]
\begin{center}
\begin{tabular}{|l|l|l|}
\hline Name                            & Axiom                                                                         &Condition                              
\\ \hline

$\mb{I}_{{\mt{P,R}}}^{\mt{f}}$ & $\B \vp   \et \B\n  \psi    \jt \n\B\n( \vp  \jt  \psi  ) \et \n( \vp  \jt  \psi  )$                    & Unknown                                                           \\ \hline
$\mb{I}_{{\mt{R,P}}}^{\mt{t}}$ & $\B\n \vp  \et \B  \psi    \jt \n\B( \vp  \jt  \psi  ) \et ( \vp  \jt  \psi  )$                         & Unknown                                            \\ \hline
$\mb{I}_{{\mt{P,P}}}^{\mt{t}}$ & $\B \vp  \et \B  \psi  \jt \n\B( \vp  \jt  \psi  ) \et ( \vp  \jt  \psi  )$                             & Unknown                                  \\ \hline
$\mb{I}_{{\mt{R,R}}}^{\mt{t}}$ & $\B\n \vp  \et \B\n  \psi    \jt \n\B( \vp  \jt  \psi  ) \et ( \vp  \jt  \psi  )$                       & Unknown                                              \\ \hline
\end{tabular}
\end{center}
\caption{Kripke harder conditions}
\end{table}

\section{Regaining \textbf{T-BAT}}
In this section we provide an axiomatization of \textbf{T-BAT} different than the one of $\mathbf{S4}^-$ from \cite{Omori2016}.\footnote{Note, that if the original formulation of \textbf{T-BAT} in \cite{Pawlowski2022} is without the typo in the truth-table of implication that it is still quite easy to get the axiomatization of the original presentation by using our framework.}\footnote{Systems quite similar to \textbf{T-BAT} have been already studied in \cite{Coniglio2015,Coniglio2016,Ivlev1988,Coniglio2019}.}

Axiomatically \textbf{T-BAT} is characterized as the logic \textbf{W} plus the following axioms:
\begin{multicols}{2}
\begin{description}[itemsep=0pt, leftmargin=2em, style=sameline, font={\mdseries\rmfamily}]
      \item[\textsf{N1}] $\B\vp \jt \B\n\n\vp$.
      \item[\textsf{N4}] $\n\B\vp\et \vp \jt \n\B\n\n\vp$.
      \item[\textsf{B1}] $\B\vp \jt \B\B\vp$.
      \item[\textsf{B7}] $\B\n\vp \jt  \B\n\B\vp$.

      \item[$\mb{I}_{{\mt{P,P}}}^{\mt{P}}$] $\B \vp  \et \B  \psi   \jt \B( \vp  \jt  \psi  )$ 
      
\item[$\mb{I}_{{\mt{t,P}}}^{\mt{P}}$] $\n\B \vp \et \vp \et \B \psi \jt \B( \vp \jt \psi )$
\item[$\mb{I}_{{\mt{f,P}}}^{\mt{P}}$] $\n\B\n \vp \et \n \vp \et \B \psi \jt \B( \vp \jt \psi )$
\item[$\mb{I}_{{\mt{R,R}}}^{\mt{P}}$] $\B\n \vp \et \B\n \psi \jt \B( \vp \jt \psi )$
\item[$\mb{I}_{{\mt{R,t}}}^{\mt{P}}$] $\B\n \vp \et \n\B \psi \et \psi \jt \B( \vp \jt \psi )$
\item[$\mb{I}_{{\mt{R,P}}}^{\mt{P}}$] $\B\n \vp \et \B \psi \jt \B( \vp \jt \psi )$
\item[$\mb{I}_{{\mt{R,f}}}^{\mt{P}}$] $\B\n \vp \et \n\B\n \psi \et \n \psi \jt \B( \vp \jt \psi )$

\item[$\mb{I}_{{\mt{P,t}}}^{\mt{t}}$] $\B \vp \et \n\B \psi \et \psi \jt \n\B( \vp \jt \psi ) \et ( \vp \jt \psi )$
\item[$\mb{I}_{{\mt{f,R}}}^{\mt{t}}$] $\n\B\n \vp \et \n \vp \et \B\n \psi \jt \n\B( \vp \jt \psi ) \et ( \vp \jt \psi )$

\item[$\mb{I}_{{\mt{P,f}}}^{\mt{f}}$] $\B \vp \et \n\B\n \psi \et \n \psi \jt \n\B\n( \vp \jt \psi ) \et \n( \vp \jt \psi )$
\item[$\mb{I}_{{\mt{t,f}}}^{\mt{f}}$] $\n\B \vp \et \vp \et \n\B\n \psi \et \n \psi \jt \n\B\n( \vp \jt \psi ) \et \n( \vp \jt \psi )$
\item[$\mb{I}_{{\mt{t,R}}}^{\mt{f}}$] $\n\B \vp \et \vp \et \B\n \psi \jt \n\B\n( \vp \jt \psi ) \et \n( \vp \jt \psi )$

\item[$\mb{I}_{{\mt{P,R}}}^{\mt{R}}$] $\B \vp \et \B\n \psi \jt \B\n( \vp \jt \psi )$

\end{description}
\end{multicols}

We are going to use $\Gamma \vdash_{\textbf{T-BAT}} \vp$ to denote \textbf{T-BAT} provability.

\begin{thm}[Soundness]
For any $\Gamma,\vp$ we have  if $\Gamma \vdash_{\emph{\textbf{T-BAT}}} \vp$, then $\Gamma \vDash_{\emph{\textbf{T-BAT}}} \vp$ 
\end{thm}
\begin{proof}
Induction on the complexity of $\vp$.
\end{proof}

\begin{thm}[Completeness]
For any $\Gamma,\vp$ we have  if $\Gamma \vDash_{\emph{\textbf{T-BAT}}} \vp$, then $\Gamma \vdash_{\emph{\textbf{T-BAT}}} \vp$ 
\end{thm}
\begin{proof}
To prove this, we need to modify the proof of the valuation lemma. More concretely, we need to check the new cases that correspond to changes in the interpretations of connectives.

\begin{enumerate}

\item{\textbf{Negation}}
\begin{enumerate}

\item Assume $\Theta_\Gamma(\vp)= \mt{P}$, we need to show $\Theta_\Gamma(\n\vp)= \mt{R}$. So, from $\Gamma \vDash_{\mathbf{T-BAT}} \B\vp$ we need to  show $\Gamma \vDash_{\mathbf{T-BAT}} \B\n\n\vp$. This is indeed the case in the light of \textsf{N1}.
\item Assume $\Theta_\Gamma(\vp)= \mt{t}$, we need to show $\Theta_\Gamma(\n\vp)= \mt{f}$.  So, from $\Gamma \vDash_{\mathbf{T-BAT}} \n\B\vp,\Gamma \vDash_{\mathbf{T-BAT}} \vp$ we have to show that $\Gamma \vDash_{\mathbf{T-BAT}} \n\B\n\n\vp,\Gamma \vDash_{\mathbf{T-BAT}} \n\n\vp$ which is the case by \textsf{N4}.
\end{enumerate}

\item{\textbf{Modality}}
\begin{enumerate}

\item Assume $\Theta_\Gamma(\vp)= \mt{P}$, we need to show $\Theta_\Gamma(\B\vp)= \mt{P}$.  So, from $\Gamma \vDash_{\mathbf{T-BAT}} \B\vp$ we have to show $\Gamma \vDash_{\mathbf{T-BAT}} \B\B\vp$ which is indeed the case by the axiom \textsf{B1}.
\item Assume $\Theta_\Gamma(\vp)= \mt{R}$, we need to show $\Theta_\Gamma(\B\vp)= \mt{R}$.  So, from $\Gamma \vDash_{\mathbf{T-BAT}} \B\n\vp$ we have to show $\Gamma \vDash_{\mathbf{T-BAT}} \B\n\B\vp$. This follows by the validity of \textsf{B7}.
\item Assume  $\Theta_\Gamma(\psi)= \mt{P}$, we need to show $\Theta_\Gamma(\vp \jt \psi)= \mt{P}$. So, from $\Gamma \vDash_{\mathbf{T-BAT}} \B\psi$, we need to show $\Gamma \vDash_{\mathbf{T-BAT}} \B(\vp\jt \psi )$. There are four sub-cases to consider here depending on the value of $\vp$ but in all of them we get $\Gamma \vDash_{\mathbf{T-BAT}} \B(\vp\jt \psi )$ by $\mb{I}_{{\mt{P,P}}}^{\mt{P}}, \mb{I}_{{\mt{P,t}}}^{\mt{P}},\mb{I}_{{\mt{P,f}}}^{\mt{P}},\mb{I}_{{\mt{P,R}}}^{\mt{P}}$.
\item Assume  $\Theta_\Gamma(\vp)= \mt{R}$, we need to show $\Theta_\Gamma(\vp \jt \psi)= \mt{P}$. So, from $\Gamma \vDash_{\mathbf{T-BAT}} \B\n\vp$ we have to show $\Gamma \vDash_{\mathbf{T-BAT}} \B(\vp\jt \psi )$. 

\end{enumerate}

\item{\textbf{Conditional}}
\begin{enumerate}

\item Assume  $\Theta_\Gamma(\vp)= \mt{t},\Theta_\Gamma(\psi)= \mt{f}$, we need to show $\Theta_\Gamma(\vp \jt \psi)= \mt{f}$ So, from $\Gamma \vDash_{\mathbf{T-BAT}} \n\B\vp, \Gamma \vDash_{\mathbf{T-BAT}} \vp,\Gamma \vDash_{\mathbf{T-BAT}} \n\B\n\psi,\Gamma \vDash_{\mathbf{T-BAT}} \n\psi$ we need to show $\Gamma \vDash_{\mathbf{T-BAT}} \n\B\n(\vp \jt \psi),\Gamma \vDash_{\mathbf{T-BAT}} \n(\vp \jt \psi)$, which we do get by $\mb{I}_{{\mt{t,f}}}^{\mt{f}}$.

\item Assume  $\Theta_\Gamma(\vp)= \mt{f},\Theta_\Gamma(\psi)= \mt{R}$, we need to show $\Theta_\Gamma(\vp \jt \psi)= \mt{t}$ So, from $\Gamma \vDash_{\mathbf{T-BAT}} \n\B\n \vp, \Gamma \vDash_{\mathbf{T-BAT}} \vp, \Gamma \vDash_{\mathbf{T-BAT}} \B\n\psi$ we need to show $\Gamma \vDash_{\mathbf{T-BAT}} (\vp \jt \psi),\Gamma \vDash_{\mathbf{T-BAT}} \n\B(\vp \jt \psi)$, which we do get by $\mb{I}_{{\mt{f,R}}}^{\mt{t}}$.

\item Assume  $\Theta_\Gamma(\vp)= \mt{t},\Theta_\Gamma(\psi)= \mt{R}$, we need to show $\Theta_\Gamma(\vp \jt \psi)= \mt{f}$ So, from $\Gamma \vDash_{\mathbf{T-BAT}} \n\B\vp, \Gamma \vDash_{\mathbf{T-BAT}} \vp,\Gamma \vDash_{\mathbf{T-BAT}} \B\n\psi$ we need to show $\Gamma \vDash_{\mathbf{T-BAT}} \n\B\n(\vp \jt \psi),\Gamma \vDash_{\mathbf{T-BAT}} \n(\vp \jt \psi)$, which we do get by $\mb{I}_{{\mt{t,R}}}^{\mt{f}}$.
\item Assume  $\Theta_\Gamma(\vp)= \mt{P},\Theta_\Gamma(\psi)= \mt{t}$, we need to show $\Theta_\Gamma(\vp \jt \psi)= \mt{t}$ So, from $\Gamma \vDash_{\mathbf{T-BAT}} \B\vp,\Gamma \vDash_{\mathbf{T-BAT}} \n\B\psi,\Gamma \vDash_{\mathbf{T-BAT}} \psi$ we need to show $\Gamma \vDash_{\mathbf{T-BAT}} \n\B(\vp \jt \psi),\Gamma \vDash_{\mathbf{T-BAT}} (\vp \jt \psi)$. This we get by the validity  of $\mb{I}_{{\mt{P,t}}}^{\mt{t}}$.

\item Assume  $\Theta_\Gamma(\vp)= \mt{P},\Theta_\Gamma(\psi)= \mt{f}$, we need to show $\Theta_\Gamma(\vp \jt \psi)= \mt{f}$ So, from $\Gamma \vDash_{\mathbf{T-BAT}} \B\vp, \Gamma \vDash_{\mathbf{T-BAT}} \n\B\n\psi, \Gamma \vDash_{\mathbf{T-BAT}} \n\psi$ we need to show that $\Gamma \vDash_{\mathbf{T-BAT}} \n\B\n(\vp \jt \psi),\Gamma \vDash_{\mathbf{T-BAT}} \n(\vp \jt \psi)$, which we do get by $\mb{I}_{{\mt{P,f}}}^{\mt{f}}$.

\item Assume  $\Theta_\Gamma(\vp)= \mt{P},\Theta_\Gamma(\psi)= \mt{R}$, we need to show $\Theta_\Gamma(\vp \jt \psi)= \mt{R}$ So, from $\Gamma \vDash_{\mathbf{T-BAT}} \B\vp, \Gamma \vDash_{\mathbf{T-BAT}} \B\n\psi$ we need to get to $\Gamma \vDash_{\mathbf{T-BAT}} \B\n(\vp\jt \psi )$ which we do by $\mb{I}_{{\mt{P,R}}}^{\mt{R}}$.
\end{enumerate}
\end{enumerate}

The rest of the proof follows the standard strategy.
\end{proof}

Since this axiomatization is not the nicest one we are going to simplify it a bit by observing the following:
\begin{fct}
Axiom $\mb{I}_{{\mt{P,t}}}^{\mt{t}}$ and axiom \textsf{K} are equivalent over \textbf{W}, i.e.
$\vDash_{\mathsf{W}} \mb{I}_{{\mt{P,t}}}^{\mt{t}} \jt (\B(\vp \jt \psi) \jt (\B\vp \jt \B\psi)) $ and $\vDash_{\mathsf{W}}   (\B(\vp \jt \psi) \jt (\B\vp \jt \B\psi))\jt \mb{I}_{{\mt{P,t}}}^{\mt{t}} $ 
\end{fct}

\begin{proof}
Suppose there is a valuation $v$ such that  $v((\B(\vp \jt \psi) \jt (\B\vp \jt \B\psi))\in \ov{\mt{D}}$. This implies that $v(\B(\vp \jt \psi))=\mt{P}$ and so  $v(\vp \jt \psi)=\mt{P}$. Additionally, we have $v(\B\vp)=\mt{P}$, thus $v(\vp)=\mt{P}$, and $v(\B\psi)\in \ov{\mt{D}}$, so $v(\n\B\psi)\in \mt{D}$. Since $v(\vp \jt \psi)\in \mt{D}$ and $v(\vp)=\mt{P}$ we know that $v(\psi)\in \mt{D}$. From the fact that $v(\B\psi)\in \ov{\mt{D}}$ we have $v(\psi)=\mt{t}$. This means that $v(\B\vp \et \n\B\psi \et \psi),v(\vp\jt\psi), \in \mt{D}$ and $v(\n\B(\vp \jt \psi))\in \ov{\mt{D}}$. This lead us to $v(\mb{I}_{{\mt{P,t}}}^{\mt{t}})\in \ov{\mt{D}}$.

For the other direction, assume that there is a valuation $v(\B \vp   \et \n\B  \psi  \et \psi  \jt \n\B( \vp  \jt  \psi ) \et (\vp  \jt  \psi ))\in \ov{\mt{D}}$. So, $v(\B\vp)=v(\vp)=\mt{P},v(\n\B\psi)\in \mt{D},v(\psi)\in \mt{D}$ and either $v(\n\B(\vp \jt \psi))\in \ov{\mt{D}}$ or $v(\vp\jt\psi) \in \ov{\mt{D}}$. The latter is impossible, since $v(\psi)\in \mt{D}$ so we get $v(\n\B(\vp \jt \psi))\in\ov{\mt{D}}$. This implies $v(\B(\vp\jt\psi))=\mt{P}$. Put everything together we have $v(\B(\vp \jt \psi))=\mt{P}=v(\B\vp),v(\B\psi)\in \ov{\mt{D}}$, so $v(\B(\vp \jt \psi) \jt (\B\vp \jt \B\psi))\in \ov{\mt{D}}$. 
\end{proof}

\section{Conclusions}

In this paper, we provide a proper axiomatization of \textbf{T-BAT} consequence relation. Additionally, we have studied all the extensions of this system, noting that not all of them are sublogics of \textbf{S4} or even of \textbf{S5}. One of the open philosophical problems is to provide a more substantial characterization of the notion of informal provability. Currently, this notion does not seem to be fully operational, and it is difficult to judge the correctness of various inference patterns that would allows us to determine whether any of the particular extensions of \textbf{T-BAT} could be potentially use for informal provability. However, if such a case could be made, it would imply that the proper logic of informal probability may not be a sublogic of \textbf{S5} and this in turn would lead to intriguing philosophical consequences.\footnote{This reserach has been supported by the FWO senior-postdoctoral grant 1255724N. The language of this paper has been refined by an LLM.}
\bibliographystyle{plain}
\bibliography{/home/haptism/Dropbox/Additional_stuff/Bibliographyfile/haptism}

\end{document}